\documentclass[reqno]{amsart}

\usepackage{amssymb}
\usepackage[mathscr]{euscript}
\usepackage{amsfonts}
\usepackage{amsmath}
\usepackage{amsthm}
\usepackage{graphicx}

\makeatletter \@addtoreset{equation}{section} \makeatother

\newtheorem{theorem}{Theorem}[section]
\newtheorem{lemma}[theorem]{Lemma}
\newtheorem{proposition}[theorem]{Proposition}
\newtheorem{corollary}[theorem]{Corollary}

\newcommand{\mf}[1]{{\mathfrak #1}}
\newcommand{\mb}[1]{{\mathbf #1}}
\newcommand{\bb}[1]{{\mathbb #1}}

\renewcommand\thefigure{\thesection.\@arabic\c@figure}
\renewcommand\thetable{\thesection.\@arabic\c@table}

\newcommand{\real}{\mathbb{R}}

\newcommand{\<}{\langle}
\renewcommand{\>}{\rangle}

\numberwithin{equation}{section}

\begin{document}

\title[Growth dynamics evolving according to the Gompertz curve]
{A Markovian growth dynamics on rooted binary trees evolving
  according to the Gompertz curve}

\author{C. Landim, R. D. Portugal and B. F. Svaiter}

\address{\noindent IMPA, Estrada Dona Castorina 110, 22460-320
    Rio de Janeiro, Brazil and CNRS UMR 6085, Universit\'e de Rouen,
    Avenue de l'Universit\'e, BP.12, Technop\^ole du Madril\-let,
    F76801 Saint-\'Etienne-du-Rouvray, France. 
\newline e-mail: \rm \texttt{landim@impa.br}}  

\address{\noindent Faculty of Medicine, Federal University of Rio
    de Janeiro} 

  \address{\noindent IMPA, Estrada Dona Castorina 110, 22460-320 Rio
    de Janeiro, Brazil.  Partially supported by
    CNPq grants  302962/2011-5 and  474944/2010-7, FAPERJ grant 
   E-26/102.940/2011 and by PRONEX-Optimization
\newline e-mail: \rm \texttt{benar@impa.br} }


\begin{abstract}
  Inspired by biological dynamics, we consider a growth Markov process
  taking values on the space of rooted binary trees, similar to the
  Aldous-Shields model \cite{AS}. Fix $n\ge 1$ and $\beta>0$. We start
  at time $0$ with the tree composed of a root only. At any time, each
  node with no descendants, independently from the other nodes,
  produces two successors at rate $\beta(n-k)/n$, where $k$ is the
  distance from the node to the root. Denote by $Z_n(t)$ the number of
  nodes with no descendants at time $t$ and let $T_n = \beta^{-1} n
  \ln( n /\ln 4) + (\ln 2)/(2 \beta)$. We prove that $2^{-n} Z_n(T_n +
  n \tau)$, $\tau\in\bb R$, converges to the Gompertz curve $\exp (-
  (\ln 2) \, e^{-\beta \tau})$. We also prove a central limit theorem
  for the martingale associated to $Z_n(t)$.
\end{abstract}

\maketitle

\section{Introduction}
\label{sec:intr}

Gompertz model was originally proposed as an actuarial
curve~\cite{Gompertz} to model mortality of an aging population, and
this approach still has applications in current survival analysis
\cite{Coe}. One century after its creation, the model surpassed its
original realm and was used as a biological growth curve \cite{Wright,
  Winsor}.  Since then, this function has successfully described
animal growth~\cite {Laird1}, regeneration \cite{Wallenstein}, and
tumor growth after the pioneering work of Laird~\cite {Laird3}.

The astonishing fact is not only that Gompertz model fits successfully
in many cases of biological growth, but mainly that its biological
foundation is still not fully understood and that the Gompertz curve
has not yet been derived as the scaling limit of some microscopic
dynamics.

Inspired by recent biological experiments on cellular aging, we
present in this article a microscopic Markovian growth dynamics which
leads to the Gompertz curve.  Telomeres are nucleo-proteins located at
the end terminal of chromosomes which shorten at each somatic cell
division. In cultured cells, growth is not observed indefinitely, the
division rate slows down and ultimately ceases \cite{Hayflick}. There
are significant evidences that telomeres act as a molecular counting
device which regulates the number of cell divisions and limits further
division after a critic length is achieved \cite{Harley}.  Recent
experimental evidences suggest also that telomere shortening may also
be related with the decrease of mitotic rate, the proportion of cells
in a tissue that are undergoing mitosis \cite{Baxter}.

On the mathematical side, computer simulation of a very simple
discrete-time stochastic model of telomere regulated growth yielded
growth curves similar to the Gompertzian model \cite{Portugal}.  The
basic assumption of this model was a \emph{linear} decrease in mitosis
probability with telomere shortening.  We take a step further in this
article by studying a continuous-time branching process. Instead of
simulations, we prove that the size of the population properly
rescaled in time converges to the Gompertz curve. We also estimate the
mean telomere size of the cell population, a quantity that can be
actually measured in cultured cells.

Our model is similar to the Aldous-Shields model \cite{AS, DM, BP} for
a rooted, growing random binary tree. While in their model each active
vertex becomes inactive and activates its two descendant at a rate
which decreases exponentially with the distance of the vertex to the
root, in our model the rate decreases linearly. 

\section{The stochastic model}
\label{sec:model}

We present in this article a random growth model which we will
interpret as a cell division process regulated by telomere shortening.
Cell division is clearly not an instantaneous process since after a
division each cell must synthesize a number of cellular components
before it can divide again. We will assume, however, that these
processes occur in a time scale much smaller than the time scale in
which cell division takes place.

Assume that the initial state is a single cell with telomere length
$L_0$, that a \emph{fixed} amount of basis, say $\delta$, is lost by
each telomere at each cell division, and that a cell reaches mitotic
senescence, which means that it does not divide anymore, when its
telomere attains a critical length $L_{\min}$. Without loss of
generality, we may suppose that $n= (L_0-L_{\min})/\delta$ is a
positive integer representing the maximum number of divisions which a
cell may undergo, the so called Hayflick limit in the biological
literature.

The total time spent in a cell cycle and arrest for a cell which has
undergone $k$ divisions, $0\le k\le n$, is modeled as an exponential
random variable with parameter $\beta (n-k)/n$, where $\beta>0$ is a
fixed parameter representing the rate at which a cell with telomere of
length $L_0$ undergoes a division. These exponential random variables
are of course assumed to be mutually independent.

The dynamics just described is a Markov process taking values on the
space of rooted binary trees, similar to the Aldous-Shields model
\cite{AS}. We start at time $0$ with the tree composed of a root
only. At any time, each node with no descendants, independently from
the other nodes, produces two successors at a rate equal to
$\beta(n-k)/n$, where $k$ is the distance from the node to the root.
Eventually the process reaches the rooted binary tree with $n$
generations, which is the absorbing state for the process. 

Let $X(k,t) = X_n(k,t)$, $0\le k\le n$, be the number of cells in the
population which has undergone \emph{exactly} $k$ mitosis at time
$t$. In the tree formulation, $X(k,t)$ represents the number of nodes
in the $k$th generation with no descendants, called the active nodes
of the $k$th generation. The cell population consists therefore of the
active nodes.  It follows from the previous assumptions that
$X(t)=(X(0,t), \dots, X(n,t))$ is a Markov chain in the state space
\begin{equation*}
\Omega = \{ (x_0,\dots,x_n) \in \mathbb{Z}_+^{n+1}\;|\; 
\sum_{i=0}^n 2^{n-i}x_i=2^n \}
\end{equation*}
with generator $Q$ given by
\begin{equation}
\label{eq:markov.00}
(Qg)(\mb x) := \sum_{k=0}^{n-1} \lambda_k \, x_k 
\big[ \, g(T_k \mb x) - g( \mb x ) \, \big] \;,
\end{equation}
where $\mb x = (x_0,\dots,x_n)$, $T_k \mb x = (x_0,\dots, x_k-1,
x_{k+1}+2, \dots ,x_n)$, $\lambda_k = \beta (n-k)/n$, $0\le k\le n-1$,
and $g: \Omega\to \bb R$ is a generic function.

Denote by $\bb X_n(t)$ the size of the population at time $t$
renormalized by $2^n$:
\begin{equation*}
\bb X_n(t) \;=\; 2^{-n} \sum_{k=0} ^{n} X(k,t)\;,
\end{equation*}
by $x_k(t)$, $0\le k\le n$, $t\ge 0$, the expected number of cells
which undergone $k$ mitosis at time $t$:
\begin{equation*}
x_k(t) = E \big[ X(k,t)\, \big] \quad\text{and let}\quad  
{\mathbf x}(t)=\begin{bmatrix} x_0(t),\cdots,x_n(t)\end{bmatrix}^T\;. 
\end{equation*}
Denote by $S_n(t)$ the expected number of cells:
\begin{equation}
\label{02}
S_n(t) \;:=\; 2^{-n} \sum_{i=0} ^{n} x_i(t)\;.
\end{equation}
As before, in the tree formulation, $x_k(t)$ represents the expected
number of active nodes in the $k$-th generation and $S_n(t)$ the
expected number of active nodes, which corresponds to the expected
size of the inner boundary of the tree.  Let
\begin{equation}
\label{03}
T_n \;=\; \frac n\beta \ln\Big( \frac n {\ln 4} \Big) \; +\;
\frac{\ln 2}{2 \beta}\; \cdot
\end{equation}
We show in \eqref{01} that $T_n$ is at distance $O(n^{-1})$ from the
time at which the expected number of active nodes is equal to one half
of the final size.

As in \cite{AS}, consider the progress of the tree along the leftmost
branch. Let $H^n_k$ the time of the $k$-th division, so that $H^n_k =
\sum_{0\le i <k} \lambda_i^{-1} \mf e_i$, where $\mf e_i$ are
i.i.d. mean one exponential random variables. In particular, $E[H^n_k]
= (n/\beta) \sum_{n-k+1\le i\le n} i^{-1} \approx (n/\beta) \ln
[n/(n-k+1)]$, which means that that $H^n_k \approx T_n$ only for
$k\approx n$.  Hence, in the time scale $T_n$ a cell is close to
senescence.

Let 
\begin{equation*}
\mb X_n(\tau) \;=\; \bb X(T_n+n\tau) \;, \quad \tau\in\bb R\;,
\end{equation*}
be the total population rescaled in size and time. The main result of
this article states that for each $\tau\in\bb R$, $\mb X_n(\tau)$
converges in probability as $n\uparrow\infty$:

\begin{theorem}
\label{s00}
The process $\mb X_n(\tau)$, $\tau\in\bb R$, converges in distribution
to the Dirac measure concentrated on the function $\bb G(\tau) = \exp
(- (\ln 2) \, e^{-\beta \tau} )$, the Gompertz curve with growth rate
$\beta$.  In particular, for every $\tau\in\bb R$, $\mb X_n(\tau)$
converges in probability, as $n\uparrow\infty$, to $\bb G(\tau)$:
\begin{equation*}
\lim_{n\to\infty} \mb X_n(\tau) \;=\; \bb G(\tau) \;=:\;
\exp \big(- (\ln 2) \, e^{-\beta \tau} \big)\;.
\end{equation*}
\end{theorem}

We also prove a central limit theorem for the martingale associated to
$\mb X_n(\tau)$.  For $\tau\ge 0$ and $n\ge 1$, let $M_n$ be the
martingale given by
\begin{equation*}
\begin{split}
& M_n(\tau) \;=\; 2^{n/2} \Big( \mb X_n(\tau) - \mb X_n(0) 
- \int_0^\tau V_n (s) \, ds \Big) \;, \\
&\quad \text{where}\quad 
V_n (s) \;=\; n (Q \bb X_n) (T_n + ns)
\;=\; \frac n{2^n} \, \sum_{j=0}^n \lambda_j 
\, X(j,T_n + ns)\;.
\end{split}
\end{equation*}

\begin{theorem}
\label{s06}
The martingale $M_n(\tau)$, $\tau\ge 0$, converges in distribution to
$B_{v(\tau)}$, where $B$ is a Brownian motion and $v(\tau) = \exp (-
(\ln 2) \, e^{-\beta \tau}) - (1/2)$.
\end{theorem}

The proof of Theorem \ref{s00} relies on an explicit expression for
the expected number of cells. Let $W_n$ be the rescaled growth curve
\begin{eqnarray*} 
W_n(\tau) \;=\; S_n(T_n+n\,\tau) \;=\; E \big[ \mb X_n(\tau) \big]
\;, \quad \tau\in \bb R\;.
\end{eqnarray*}
The rescaled growth curve $W_n$ converges pointwisely (and therefore
uniformly) to the Gompertz curve with growth rate $\beta$. 

\begin{proposition}
\label{s02}
For every $t>0$ and every $\tau\in\bb R$,
\begin{equation*}
S_n(t) \;=\;\left(1-\frac{e^{-\beta t/n}}{2}\right)^n\;,
\quad  
\lim_{n\to\infty} W_n(\tau) \;=\; 
\exp \big(- (\ln 2) \, e^{-\beta \tau} \big)\;.
\end{equation*}
\end{proposition}

The proof of this result provides in fact an expansion in $n^{-1}$ of
$\ln W_n(\tau)$.  We may also compute the asymptotic behavior of the
time derivative of the rescaled growth curve $W_n$.  A simple
computation with the generator gives that
\begin{equation}
\label{06}
\frac d{dt} S_n(t) \;=\; \frac{\beta}{2^n} 
\sum_{i=0}^{n} \frac{(n-i)}n \, x_i(t)\;.
\end{equation}
We show in Section \ref{sec:tt} that
\begin{equation}
\label{05}
(\frac d{dt} \ln S_n)(t) \;=\; 
\frac{\beta\, e^{-\beta t/n}}{2-e^{-\beta t/n}}\;, \quad n\ge 1\;,
\;\; t>0\;.
\end{equation}
The next result follows from this identity and Proposition \ref{s02}.

\begin{proposition}
\label{s01}
For every $\tau\in\bb R$,
\begin{equation*}
\lim_{n\to\infty} (\frac d{d\tau} \ln W_n)(\tau)
\;=\;  (\ln 2) \, \beta \, e^{-\beta \tau}\;.
\end{equation*}
\end{proposition}

This result has an interpretation in our biological model.  In
cultured cells, telomere length is not evaluated
individually. Instead, what is actually measured is the \emph{mean}
telomere length of a colony of cells ~\cite{Baxter}. Therefore, in
order to verify the fit of the stochastic model to real data, we must
obtain the expected mean telomere length predicted by the model.

The asymptotic telomere length may be estimated in two different
regimes. The telomere length of a cell which has undergone $k$ mitosis
is $L_{\min}+(n-k)\delta$.  Hence, the total telomere length at
time $t$ is $\sum_{0\le k\le n} [L_{\min}+(n-k)\delta]\, X(k,t)$.  Let
$\ell_n(t)$ be the expectation of the total telomere length and let
$L_n(t)$ be the average telomere length:
\begin{equation*}
\ell_n(t) \;=\; E\Big[ \sum_{k=0}^n
\big[L_{\min} + (n-k)\, \delta \big]\, X(k,t) \Big]\; ,
\quad L_n(t) \;=\; \frac{\ell_n(t)}{2^n S_n(t)}\;\cdot
\end{equation*}
Proposition \ref{s03} below follows from identity \eqref{05} and
Proposition \ref{s01}. Its proof is presented at the end of Section
\ref{sec:tt}.

\begin{proposition}
\label{s03}
For every $t\ge 0$, $\tau\in\bb R$,
\begin{equation*}
\begin{split}
& L_n(t n) \;=\; L_{\min} \;+\; (L_0-L_{\min})
\frac{e^{-\beta t}}{2-e^{-\beta t}}\;, \quad n\ge 1\;, \\
&\quad \lim_{n\to\infty} n \, \{ L_n(T_n + \tau n) - L_{\min}\}
\;=\;  (\ln 2) \, (L_0-L_{\min}) \, e^{-\beta \tau}\;.
\end{split}
\end{equation*}
\end{proposition}

\section{The dynamics of the stochastic model}
\label{sec:din}

We prove in this section Proposition \ref{s02}.  By Kolmogorov's
forward equation
\begin{equation}
\label{f01}
\frac{d}{dt}{\mathbf x} (t) \;=\; \mathbf{M} \, \mb  x (t) \;,
\end{equation}
where $\mathbf{M}$ is the square matrix with entries $m_{i,j}$, $0\le
i,j \le n$, given by
\begin{equation}
\label{eq:def.m}
m_{i,j}=
\begin{cases}
-\lambda_{i}, &j=i,\\
2\lambda_{i-1}, &j=i-1, \\
0&\mbox{otherwise}   
\end{cases}
\end{equation}

The solution of the previous linear ordinary differential equation
with initial condition $\mb x (0)= \begin{bmatrix} 1, 0, \cdots,
  0\end{bmatrix}^T$ is ${\mathbf x} (t)=\exp(t\mathbf{M})\, \mb x
(0)$, $t\ge 0$.  Moreover, the expected number of cells at time $t$
can be calculated as the matrix product of the vector $[1,1,\dots,1]$
with $\mb x(t)$:
\[
\sum_{i=0} ^n x_i(t) \;=\; [1,1,\dots,1] \, {\mathbf x} (t)
\;=\; [1,1,\dots,1] \exp(t\mathbf{M})\, \mb x (0) \;.
\]
 
The proof of Proposition \ref{s01} relies on the analysis of the
matrix $\mb M$, presented in Lemma \ref{lm:eig} and in Corollary
\ref{cr:aux} below.  Since $\mb M$ is lower-triangular, the spectrum
of $\mb M$ is $-\lambda_0, -\lambda_1, \cdots, -\lambda_n$. For
$\gamma\in\real$ and $0\le i,k \le n$, let
\begin{equation}
  \label{eq:def.eig.basis}
  a(i,k,\gamma)=
  {a(\gamma)}_{i,k} =
\begin{cases}
 0 & i<k,\\
\displaystyle  \gamma^{i-k}\binom{n-k}{i-k}& i\geq k \;. \\
\end{cases}
\end{equation}

\begin{lemma}
\label{lm:eig}
The vector $\mathbf{w}_k\in \real^{n+1}$, $0\le k\le n$, defined by
\begin{equation}
\label{eq:def.eig.k}
\mathbf{w}_k=
\begin{bmatrix}
w_{0,k}, w_{1,k},\cdots, w_{n,k}
\end{bmatrix}^T \;, \qquad w_{i,k}={a(-2)}_{i,k} \;,
\end{equation}
is a right-side eigenvector of $\mathbf{M}$ corresponding to the
eigenvalues $-\lambda_{k}$.
\end{lemma}

\begin{proof}
Fix $0\le k\le n$. To prove this lemma we must evaluate
\[\mathbf{b}\;=\;
(\mathbf{M}+\lambda_{k}\mathbf{I})\, \mathbf{w}_k
\]
where $\mathbf{I}$ is the identity matrix.  Using the fact that
$\mathbf{M}$ is lower-triangular we get from the definition of $\mb
w_k$ that $b_{i}=0$ for $i<k$.  For $i=k$, since
$m_{k,k}=-\lambda_{k}$, $b_{k}=(m_{k,k}+\lambda_{k})\, w_{k,k}=0$.
This proves the lemma for $k=n$.

For $0\le k <n$ and $i>k$,
\[
b_{i} = m_{i,i-1}\;w_{i-1,k}+(m_{i,i}+\lambda_{k})\, w_{i,k}
\]
By \eqref{eq:def.m} and by the definition of $\lambda_i$, we have
\begin{eqnarray*}
b_{i}&=& 2\lambda_{i-1}\,w_{i-1,k} \;+\; (-\lambda_{i}+\lambda_{k})\, w_{i,k}\\
&=&\frac{\beta}{n} \, \big[ 2(n+1-i)w_{i-1,k} \,+\, (i-k)w_{i,k} \big]\\
&=&\frac{\beta}{n}\, 2(n+1-i)\,
\big[ w_{i-1,k} + \frac{i-k}{2(n+1-i)}w_{i,k}\big]\;.
\end{eqnarray*}
According to \eqref{eq:def.eig.basis}, for $k<i\leq n$,
\[
\frac{w_{i-1,k}}{w_{i,k}} \;=\; \frac{a(-2)_{i-1,k} }{a(-2)_{i,k}}
\;=\;\frac{i-k}{-2\, (n+1-i)}\;\cdot
\]
This proves the lemma.
\end{proof}

Define the square matrices $\mathbf{A}(\gamma)$, $\mathbf{D}$ by
\begin{equation}
\label{eq:def.ad}
\mathbf{A}(\gamma)=\{a(\gamma)_{i,j}\}\; , \quad
\mathbf{D}=\mathrm{diag}\{-\lambda_0,\dots,-\lambda_{n-1},-\lambda_n\}
\;.  
\end{equation}
In view of the previous lemma, 
\[
\mathbf{M}=\mathbf{A}(-2)\, \mathbf{D}\,\left[\mathbf{A}(-2)\right]^{-1}
\]
so that
\begin{equation}
\label{exptm}
\exp(t\mathbf{M})=\mathbf{A}(-2)\ 
\exp(t\mathbf{D})\,\left[\mathbf{A}(-2)\right]^{-1}\;, \quad t\ge 0\;.
\end{equation}
To evaluate this expression we will need two auxiliary results.

\begin{lemma}
\label{lm:aux}
For any $\eta\in \real$,
\[
[\eta^n,\eta^{n-1},\dots, \eta, 1] \, \mathbf{A}(\gamma) \;=\;
[(\eta+\gamma)^n,(\eta+\gamma)^{n-1},\dots, (\eta+\gamma), 1]\,.
\]
\end{lemma}

\begin{proof}
Let 
\[
\begin{bmatrix}
u_0,\cdots,u_{n}
\end{bmatrix}
\;=\; [\eta^n,\eta^{n-1},\dots, \eta, 1]
\, \mathbf{A}(\gamma) \;.
\] 
A direct calculation, together with \eqref{eq:def.eig.basis} yields
\begin{equation*}
u_k \;=\; \sum_{i=0} ^{n} \eta^{n-i}\; a_{i,k}(\gamma)
\;=\; \sum_{i=k} ^{n} \eta^{n-i}\; \gamma^{i-k}\;
\binom{n-k}{i-k} \;=\; (\eta+\gamma)^{n-k} \;.
\end{equation*}
This concludes the proof of the lemma.
\end{proof}

\begin{corollary}
\label{cr:aux}
For any $\gamma,\mu\in\real$, $\mathbf{A}(\gamma)
\mathbf{A}(\mu)=\mathbf{A}(\gamma+\mu)$, and
$\mathbf{A}(0)=\mathbf{I}$. In particular
$[\mathbf{A}(\gamma)]^{-1}=\mathbf{A}(-\gamma)$.
\end{corollary}

\begin{proof}
By Lemma~\ref{lm:aux}, for any $\eta\in \bb R$,
\[
[\eta^n,\eta^{n-1},\dots, \eta, 1] \, \mathbf{A}(\gamma) 
\, \mathbf{A}(\mu) \;=\;
[\eta^n,\eta^{n-1},\dots, \eta, 1]\, \mathbf{A}(\gamma+\mu)\;.
\]
The first assertion of the Corollary follows from this identity and
from the fact that the span of $\{ [\eta^n,\eta^{n-1},\dots, \eta,
1]\;:\;\eta\in \real\}$ is $\real^{n+1}$.  The second claim of the
corollary follows from the definition of $\mathbf{A}(\gamma)$.
\end{proof}

\noindent{\sl Proof of Proposition \ref{s02}.}
By \eqref{exptm} and by the previous corollary,
\begin{equation}
\label{eq:xt}
\mathbf{x}(t) \;=\; \mathbf{A}(-2)\,
\exp(t\mathbf{D})\, \mathbf{A}(2) \, \mb x(0)\;.
\end{equation}
Since $e^{-t\lambda_k}=(e^{-t\beta/n})^{n-k}$, $0\le k\le n$, by
definition of the diagonal matrix $\mb D$, 
\begin{equation*}
\exp(t\mathbf{D}) \;=\; 
\mathrm{diag}\{(e^{-t\beta/n})^n,\cdots,(e^{-t\beta/n})^1,1\}\;.
\end{equation*}
Hence,
\begin{equation}
\label{eq:to.evaluate}
\sum_{i=0} ^n x_i(t) \;=\; [1,1,\dots,1] \, \mathbf{A}(-2)\,
\exp(t\mathbf{D}) \,\mathbf{A}(2) \, \mb x(0) \;.
\end{equation}
Applying twice Lemma~\ref{lm:aux} we get that
\begin{equation*}
\begin{split}
& [1,1,\dots,1]\, \mathbf{A}(-2) \, \exp(t\mathbf{D}) \,\mathbf{A}(2)
\\ 
&\qquad = \; [(2-e^{-\beta t/n} )^n, (2-e^{-\beta t/n})^{n-1}, 
\dots, 2-e^{-\beta t/n},1]\;,  
\end{split}
\end{equation*}
so that
\begin{equation*}
\sum_{i=0} ^{n} x_i(t) \;=\; (2-e^{-\beta t/n})^n
\;=\; 2^n\;\left(1-\frac{e^{-\beta t/n}}{2}\right)^n\;.
\end{equation*}
This concludes the proof of the first assertion of the theorem.

Recall that we denote by $W_n(\tau)$ the normalized growth curve.  It
follows from the definition of $S_n(t)$ that
\begin{equation}
\label{eq:ufa}
W_n(\tau) \;=\;
\left(1-\frac{e^{-\beta \tau}\theta_n\ln 2}{n} \right)^n\;,
\end{equation}
where $\theta_n= 2^{-(1/2n)}$. It remains to let
$n\uparrow\infty$. \qed \smallskip

We conclude this section showing that $T_n$ is at distance $n^{-1}$
from the time at which the expected number of active nodes is equal to
one half of the final size.

Denote by $t_*$ the time at which the expected population size is half
of the final size:
\[
n\ln \left(1-\frac{e^{-\beta t_*/n}}{2}\right) \;=\; -\ln 2\;.
\]
We claim that
\begin{equation}
\label{01}
t_* \;=\; \frac n\beta \ln\Big( \frac n {\ln 4} \Big) \; +\;
\frac{\ln 2}{2 \beta} \; +\; O(n^{-1})\; .
\end{equation}

Indeed, dividing both sides of the penultimate formula by $n$ and
taking exponentials on both sides of this equation, we obtain that
\[
e^{-\beta t_*/n} \;=\; 2\left(1-e^{-\ln(2)/n}\right)\;.
\]
By Taylor expansion and setting $t_*=(n/\beta)\ln(n/\ln 4)+h_*$ we
obtain
\[
\frac{\ln 4}{n} \, e^{-\beta h_*/n} \;=\;
2\, \Big( \frac{\ln 2}{n}-\frac{1}{2} 
\Big(\frac{\ln 2}{n} \Big)^2+O(n^{-3})\Big)\,.
\]
Hence
\[
e^{-\beta h_*/n}\;=\; 1 \;-\; \frac{\ln 2}{2n} \;+\; O(n^{-2})\,.
\]
Taking logarithms on both sides, by Taylor expansion,
\[ 
-\frac{\beta}{n}h_* \;=\; -\frac{\ln 2}{2n} \;+\; O(n^{-2})\,.
\]
This proves claim \eqref{01}.

\section{A generating function}
\label{sec:tt}

We prove in this section Proposition \ref{s01} and identity
\eqref{05}.  Let
\begin{equation}
\label{eq:def.Psi}
\Psi(u,t)\;=\; \sum_{k=0} ^n u^{n-k} X(k,t) \;,\quad u\;, t \in\bb R_+\;,
\end{equation}
and denote by $\psi(u,t)$ be the expected value of $\Psi(u,t)$:
\begin{equation}
\label{eq:def.psi}
\psi(u,t) \;=\; E[ \Psi(u,t) ] \;=\; 
\sum_{k=0} ^n u^{n-k}  E[X(k,t)] \;=\; \sum_{k=0} ^n u^{n-k} x_k(t)\;.
\end{equation}
Hence, by \eqref{eq:xt},
\begin{equation*}
\begin{split}
\psi(u,t) \; &=\; [u^n, u^{n-1},\cdots,u,1]\;\mathbf{x}(t) \\
\;& =\; [u^n, u^{n-1},\cdots,u,1]\,\mathbf{A}(-2)\,\exp(t\mathbf{D})\,
\mathbf{A}(2) \mathbf{x}(0)\;.
\end{split}
\end{equation*}
By Lemma~\ref{lm:aux} and by the formula \eqref{eq:def.ad} for the
diagonal matrix $\mathbf{D}$, 
\begin{equation}
\label{eq:psi}
\psi(u,t)\;=\; (2+(u-2)e^{-\beta t/n})^n\;.  
\end{equation}

On the other hand,
\begin{equation}
\label{07}
\begin{split}
& E\Big[\sum_{k=0} ^nX(k,t)\Big]\;=\; 
\sum_{k=0}^n x_k(t)\;=\; \psi(1,t)\;,\\
&\quad E\Big[ \sum_{k=0} ^n (n-k)X(k,t)\Big] \;=\; 
\sum_{k=0} ^n (n-k) x_k(t) \;=\; \psi_u(1,t)\;,
\end{split}
\end{equation}
where $\psi_u$ stands for the partial derivative of $\psi$ with
respect to the first variable.  

It follows from \eqref{02}, \eqref{06}, the previous equations and
\eqref{eq:psi} that
\begin{equation*}
(\frac d{dt} \ln S_n)(t) \;=\; \frac{\beta}n\, 
\frac{\psi_u(1,t)}{\psi(1,t)} \;=\;
\frac{\beta\, e^{-\beta t/n}}{2-e^{-\beta t/n}}
\end{equation*}
for $n\ge 1$, $t>0$.  This is the identity presented in \eqref{05}.

Recall the definition of $L_n(t)$, the average telomere length at time
$t$, introduced just before the statement Proposition \ref{s03}. It
follows from \eqref{eq:psi} and \eqref{07} that
\begin{align*}
\ell_n(t) \; &=\;
E\Big[  \sum_{k=0} ^n \big\{ L_{\min}+(n-k)\delta\big\}\, X(k,t) \Big] \\
& =\; L_{\min}\, (2-e^{-\beta t/n})^n \;+\;
\delta\, n \,(2-e^{-\beta t/n})^{n-1}\, e^{-\beta t/n}\; ,
\end{align*}
and
\begin{equation*}
L_n(t)\;=\; L_{\min} \;+\; (L_0-L_{\min})\,
\frac {e^{-\beta t/n}}{2-e^{-\beta t/n}}
\end{equation*}
because $\delta n = L_0-L_{\min}$.  The assertions of Proposition
\ref{s03} are a straightforward consequence of the previous formula.

\section{Proof of Theorem \ref{s00}}
\label{sec6}

We prove in this section Theorem \ref{s00} by estimating the
covariances of the process $\mb X_n(t)$. Let
\begin{equation*}
F_j(t) \;=\;  E[X(j,t)]\;, \quad 0\le j\le n\;, 
\end{equation*}
and let $F_{j,k}(t)$ be the covariance between $X(j,t)$ and $X(k,t)$:
\begin{equation*}
F_{j,k}(t) \;=\; E\big[ X(j,t)\, X(k,t) \big] \;-\;
E\big[ X(j,t) \big]\, E\big[ X(k,t) \big]\;, \quad
0\le j\le k\le n \;.
\end{equation*}
Let $F(t)$ be the column vector with $m=(n+1)(n+2)/2$ coordinates
given by 
\begin{equation*}
\begin{split}
& F(t) \;=\; \\
& \big[ F_{0,0}(t)\,,\, F_{0,1}(t)\,,\, F_{1,1}(t)\,,\, \dots\,,\,
F_{0,j}(t)\,,\, F_{1,j}(t)\,,\,\dots\,,\, F_{j,j}(t)\,,\, \dots\,,\, 
F_{n,n}(t)\big]^T \;.
\end{split}
\end{equation*}
An elementary computation shows that
\begin{equation}
\label{08}
\frac d{dt} F(t) \;=\; \Sigma F(t) \;+\; G(t)\;,
\end{equation}
where $\Sigma$ is the $m\times m$ matrix given by
\begin{equation*}
\Sigma \;=\; 
\begin{bmatrix}
M_0&  0 & 0 & 0 & 0 & 0 \\
D_1&  M_1 & 0 & 0 & 0 & 0 \\
0 & D_2 & M_2 & 0 & 0 & 0  \\
 &   & \ddots & \ddots &   &  \\     
0 & 0 & 0  & D_{n-1} & M_{n-1}  & 0 \\ 
0 & 0 & 0 & 0 & D_n & M_n  
\end{bmatrix}
\;,
\end{equation*}
$M_j$ is a square $(j+1)\times (j+1)$ matrix with entries $M_j(a,b)$,
$0\le a, b\le j$, given by
\begin{equation*}
\begin{split}
& M_j(a,a) \;=\; -(\lambda_a + \lambda_j)\;, \quad 0\le a\le j\;, 
\quad M_j(j,j-1) \;=\; 4\lambda_{j-1}\;, \\
& \quad M_j(a+1,a) \;=\; 2\lambda_a \;, \quad 0\le a\le j-2\;, 
\quad M_j(a,b) \;=\; 0 \text{ otherwise}\;:
\end{split}
\end{equation*}
\begin{equation*}
M_j \;=\;
\begin{bmatrix}
-(\lambda_0+\lambda_j)&  0 & 0 & 0 & 0 & 0 \\
2\lambda_0 & -(\lambda_1+\lambda_j)  & 0 & 0 & 0 & 0 \\
0 & 2\lambda_1 & -(\lambda_2+\lambda_j) & 0 & 0 & 0  \\
 &   & \ddots & \ddots &   &  \\     
0 & 0 & 0  & 2\lambda_{j-2}  & -(\lambda_{j-1}+\lambda_j)& 0 \\ 
0 & 0 & 0 & 0 & 4\lambda_{j-1}  & -2\lambda_j
\end{bmatrix}
\;,
\end{equation*}
and $D_j$ is the $(j+1)\times j$ matrix whose first $j$ lines form the
matrix $2\lambda_{j-1} I_j$, where $I_j$ is the $j\times j$ identity,
and whose last line has only zeros. Moreover, $G(t)$ is the vector $[
G_{0}(t)\,,\, G_{1}(t)\,,\, \dots\,,\, G_{n}(t)]^T$ and $G_j(t)$ is
the column vector with $j+1$ entries given by
\begin{equation*}
\begin{split}
& G_0(t) \;=\; \big[\lambda_0 F_0(t)\big]^T\;, \quad G_1(t)\;=\; 
\big [-2\lambda_0 F_0(t) \,,\, 4\lambda_0 F_0(t) + \lambda_1 F_1(t)\big]^T\;,\\
& \quad G_j(t)\;=\; \big[0, \dots, 0, -2\lambda_{j-1}
F_{j-1}(t) \,,\, 4\lambda_{j-1} F_{j-1}(t) + \lambda_j F_j(t)\big]^T
\;,\;\; 2\le j\le n-1\;, \\
&\qquad G_n(t)\;=\; \big[0, \dots, 0, -2\lambda_{n-1}
F_{n-1}(t) \,,\, 4\lambda_{n-1} F_{n-1}(t) \big]^T \;.
\end{split}
\end{equation*}

\begin{lemma}
\label{s08}
The matrix $\Sigma$ is diagonalizable.
\end{lemma}

\begin{proof}
The eigenvalues of $\Sigma$ are $-(\lambda_0+\lambda_j)$,
$-(\lambda_n+\lambda_j)$, $0\le j\le n$. The eigenvalues
$-(\lambda_j+\lambda_0)$, $0\le j\le n$, have multiplicity
$u_j=\lfloor j/2\rfloor +1$, and the eigenvalues
$-(\lambda_j+\lambda_n)$, $0\le j\le n$, have multiplicity $v_j=
\lfloor (n-j)/2 \rfloor +1$, where $\lfloor a\rfloor$ stands for the
integer part of $a$.

Fix $0\le j\le n$ and assume without loss of generality that $j$ is
even, $j=2k$. Consider the matrix $\Sigma_j = \Sigma +
(\lambda_0+\lambda_j) I_m$, where $I_m$ is the $m\times m$
identity. This matrix has $u_j$ zeros on the diagonal. Starting from
the bottom of the matrix, the first zero appears at the position
$(0,0)$ of the matrix $M_j$, the second one at the position $(1,1)$ of
the matrix $M_{j-1}$, and the last one at the position $(k,k)$ of the
Matrix $M_{j-k}$.

We claim that we may reduce the matrix $\Sigma_j$ to obtain a matrix
such that all lines where the matrix $\Sigma_j$ has an entry on the
diagonal equal to zero become identically equal to zero. The assertion
of the lemma follows from this claim.

This reduction is performed recursively. Let $(\ell_1, \ell_1)$ be the
position of the upmost zero in the diagonal and keep in mind that we
start counting from $0$ which means that the position $(\ell_1,
\ell_1)$ indicates in reality the $(\ell_1+1)$-th line and column.
This entry corresponds to the entry $(k,k)$ of the matrix $M_{j-k}$.
We may first reduce the matrix $\Sigma_j$ to obtain a matrix
$\Sigma_j^{(1)}$ which coincides with $\Sigma_j$ on the south-east
square matrix corresponding to the entries $\{\ell_1, \dots , m-1\}
\times \{\ell_1, \dots , m-1\}$, whose restriction to the north-west
square matrix corresponding to the entries $\{0, \dots,
\ell_1-1\}\times \{0, \dots, \ell_1-1\}$ is the identity, and whose
entries vanish on the remaining two parts of the matrix. Notice that
the $(\ell_1+1)$-th line of $\Sigma_j^{(1)}$ vanishes, and that the
entries $(i,i)$, $0\le i\le k-1$, of the Matrix $D_{j-k+1}$ also
vanish.

Denote by $(\ell_2, \ell_2)$ the coordinates of the second upmost zero
in the diagonal of the matrix $\Sigma_j$ and by $(m_2,m_2)$ the
coordinates of the first entry of the matrix $M_{j-k+1}$, so that
$m_2\le\ell_2$.  We may reduce the matrix $\Sigma_j^{(1)}$ using the
diagonal entries of the matrix $M_{j-k+1}$ above the zero entry to
obtain a new matrix $\Sigma_j^{(2)}$ with the following
properties. Only the entries with coordinates in $\{m_2, \dots, m\}
\times \{m_2, \dots, \ell_2-1\}$ have been changed. The restriction of
$M_{j-k+1}$ to the square $\{m_2, \dots, \ell_2-1\} \times \{m_2,
\dots, \ell_2-1\}$ is the identity and all entries below the diagonal
of this matrix are zeros. Note that the line of $\Sigma_j^{(2)}$
corresponding to the second upmost zero in the diagonal of the matrix
$\Sigma_j$ vanishes, and that the entries $(i,i)$, $0\le i\le k-2$, of
the Matrix $D_{j-k+2}$ also vanish. We may therefore repeat the
argument and conclude the proof of the claim.

Since the same argument applies to the eigenvalues
$-(\lambda_n+\lambda_j)$, $0\le j\le n$, the lemma is proved.
\end{proof}

\begin{proposition}
\label{s09}
For all $\tau\in\bb R$,
\begin{equation*}
\lim_{n\to\infty} \text{\rm Var} \Big[ \frac 1{2^n} 
\sum_{k=0}^n X(k, T_n + n\tau)
\Big]\;=\; 0\;.
\end{equation*}
\end{proposition}

\begin{proof}
Let $v$ be the vector $[v_0, v_1, \dots, v_n]$, where $v_0 = [1]$, and
$v_j$ is the vector with $j+1$ coordinates given by $v_j = [2, \dots,
2,1]$.  It is well known that the solution of \eqref{08} with initial
condition $F(0)=0$ is given by
\begin{equation}
\label{23}
F(t) \;=\; \int_0^t e^{(t-s)\Sigma} G(s) \, ds\;.
\end{equation}
Therefore,
\begin{equation*}
\text{\rm Var} \Big[ \frac 1{2^n} \sum_{k=0}^n X(k,t)
\Big] \;=\; \frac 1{4^n} \< v \,,\, F(t)\> 
\;=\; \frac 1{4^n} \int_0^t \< v \,,\, e^{(t-s)\Sigma} G(s)\> \, ds
\;,
\end{equation*}
where $\< \,\cdot\,,\,\cdot\,\>$ stands for the usual inner
product. 

Since, by Lemma \ref{s08}, $\Sigma$ is a diagonalizable matrix and all
its eigenvalues are negative, $\< v \,,\, e^{(t-s)\Sigma} G(s)\>$ is
absolutely bounded by $|v|\, |G(s)|$. Clearly, $|v|^2 =
(n+1)(2n+1)$. On the other hand,
\begin{equation*}
\begin{split}
|G(s)|^2 \; &=\; \big[\lambda_0 F_0(s)\big]^2 \;+\;
\big[4 \lambda_{n-1} F_{n-1}(s)\big]^2 \;+\; 
\sum_{j=0}^{n-1} \big[2 \lambda_{j} F_{j}(s)\big]^2 \\
\;&+\; \sum_{j=0}^{n-2} \big[4 \lambda_{j} F_{j}(s) 
+ \lambda_{j+1} F_{j+1}(s)\big]^2\;.
\end{split}
\end{equation*}
Since $\lambda_j\le \beta$ and each $F_j$ is positive, all terms
inside brackets are bounded by $4\beta \sum_{0\le j\le n}
F_j(s)$. Hence,
\begin{equation*}
|G(s)| \; \le \; 4 \, \beta \, \sqrt{2n+1} S_n(s) \; ,
\end{equation*}
and
\begin{equation*}
\text{\rm Var} \Big[ \frac 1{2^n} \sum_{k=0}^n X(k,T_n + n\tau)
\Big] \;\le \; \frac {4\beta\, (2n+1) \sqrt{n+1} (T_n+n\tau)}{2^n}
\; \frac {S_n(T_n+n\tau)}{2^n} \;,
\end{equation*}
because $S_n$ is an increasing function. The result follows from this
estimate and Proposition \ref{s02}.
\end{proof}

\noindent{\sl Proof of Theorem \ref{s00}.}
It follows from Proposition \ref{s02} and Proposition \ref{s09} that
the finite dimensional distributions of $\mb X_n(t)$ converge to the
finite dimensional distributions of a Dirac measure concentrated on
the Gompertz curve $\exp (- (\ln 2) \, e^{-\beta \tau})$. 

\smallskip\noindent{\bf Tightness.} To conclude the proof of Theorem
\ref{s00}, it remains to show that the process $\mb X_n$ is tight on
each finite interval $[0,T]$. Since $0\le \mb X_n(t) \le 1$ a.s. for
all $t$ and $n$, it is enough \cite{b1} to show that for any $T>0$ and
$\epsilon>0$,
\begin{equation*}
\limsup_{\delta\to 0}
\limsup_{n\to\infty} P\Big[ \, \sup \big| \mb X_n(t)
- \mb X_n(s) \big| > \epsilon \Big]\;=\; 0\;,
\end{equation*}
where the supremum is carried over all $0\le s, t\le T$ such that
$|t-s|\le \delta$. By construction, the process $\mb X_n$ is
increasing. In particular, if $t_i= i \delta$, $0\le i\le
M_\delta=\lfloor \delta^{-1} \rfloor +1$, where $\lfloor a \rfloor$
stands for the integer part of $a$,
\begin{equation*}
P\Big[ \, \sup \big| \mb X_n(t) - \mb X_n(s) \big| > \epsilon \Big]
\;\le\; M_\delta \max_{0\le i\le M_\delta} 
P\Big[ \, \big| \mb X_n(t_{i+1}) - \mb X_n(t_i) \big| > \epsilon/2
\Big]\; .
\end{equation*}
By Proposition \ref{s02}, $W_n$ converges uniformly to the Gompertz
curve, which is uniformly continuous. It is therefore enough to show
that
\begin{equation*}
\limsup_{\delta\to 0} \limsup_{n\to\infty}  M_\delta \max_{0\le i\le M_\delta} 
P\Big[ \, \big| \overline{\mb X}_n(t_{i+1}) - \overline{\mb X}_n(t_i)
\big| > \epsilon/2 \Big]\; ,
\end{equation*}
where $\overline{\mb X}_n(t) = \mb X_n(t) - W_n(t)$. By Chebychev
inequality and by Proposition \ref{s09}, for all $t\in\bb R$, $a>0$, 
\begin{equation*}
\limsup_{n\to\infty} 
P\Big[ \, \big| \overline{\mb X}_n(t)\big| > a \Big]
\;\le\; \limsup_{n\to\infty}  \frac{1}{a^2} \,
E\Big[ \overline{\mb X}_n(t)^2 \Big] \;=\; 0\;.
\end{equation*}
This concludes the proof of the tightness of $\mb X_n(t)$.

\section{Central Limit Theorem}
\label{sec7}

We prove in this last section Theorem \ref{s06}.  Recall the
definition of the martingale $M_n$ and of the process $V_n$.  We rely
in this proof on \cite[Theorem VIII.3.11]{JS}. Let $\Delta M_n(s) =
M_n(s) - M_n(s-)$. By definition of the martingale $M_n$, $\Delta
M_n(s) = 2^{-n/2} \{\mb X_n(s) - \mb X_n(s-)\}$. Since $\mb X_n(s)-
\mb X_n(s-)$ is either $0$ or $1$, for every $t\ge 0$, $\epsilon >0$,
\begin{equation}
\label{l01}
\lim_{n\to\infty} P \Big[ \sup_{s\le t} \big| \Delta M_n(s)
\big| \ge \epsilon \Big] \;=\; 0\;.
\end{equation}

Denote by $\< M_n\>_t$ the predictable quadratic variation of the
martingale $M_n(t)$.  We claim that for each $\tau\ge 0$,
\begin{equation}
\label{l02}
\< M_n\>_\tau \text{ converges in probability as $n\uparrow\infty$ to }
e^{-\ln 2 e^{-\beta \tau}} - \frac 12 \;.
\end{equation}
A straightforward computation shows that
\begin{equation*}
\< M_n\>_\tau \;=\; \frac 1{2^n} \int_{T_n}^{T_n + \tau n} \sum_{j=0}^n 
\lambda_j X(j,s) \, ds\;.
\end{equation*}
By \eqref{07}, uniformly in any compact interval of $\bb R$,
\begin{equation*}
\lim_{n\to\infty} E\Big[ \frac n{2^n} \sum_{j=0}^n \lambda_j 
X(j, T_n + \tau n)\Big] \;=\; \frac{d}{d\tau} e^{-\ln 2 e^{-\beta
    \tau}}\; .
\end{equation*}
Hence, for all $\tau\ge 0$,
\begin{equation*}
\lim_{n\to\infty} E\Big[\< M_n\>_\tau  \Big] \;=\; 
e^{-\ln 2 e^{-\beta \tau}} \,-\, \frac 12 \; \cdot 
\end{equation*}

On the other hand, by Schwarz inequality and with the notation
introduced in Section \ref{sec6},
\begin{equation*}
\begin{split}
& E\Big[ \Big( \< M_n\>_\tau -  E\big[ \< M_n\>_\tau
\big] \Big)^2 \Big] \;=\;
E\Big[ \Big( \frac 1{2^n} \int_{T_n}^{T_n + \tau n} \sum_{j=0}^n 
\lambda_j \{ X(j,s) - x_j(s) \} \, ds \Big)^2\Big] \\
& \qquad \;\le\; \frac{\tau n}{4^n} \int_{T_n}^{T_n+\tau n} 
\Big\{ 2 \sum_{j<k}\lambda_j\lambda_k F_{j,k}(s) +
\sum_{j=0}^n \lambda_j^2 F_{j,j}(s) \Big\} \, ds \;.
\end{split}
\end{equation*}
Let $w= (\lambda_0^2, 2\lambda_0\lambda_1,\lambda_1^2, \dots,
\lambda_n^2)$.  By \eqref{23}, the expression inside braces is equal
to
\begin{equation*}
\<w, F(s)\> \;=\; \int_{0}^{s} 
\big\<w, e^{(s-r) \Sigma} G(r) \big\>  \, ds\;.
\end{equation*}
At this point, we repeat the arguments of the proof of Proposition
\ref{s09} to conclude that for all $\tau\ge 0$,
\begin{equation*}
\lim_{n\to\infty} E\Big[ \Big( \< M_n\>_\tau -  E\big[ \< M_n\>_\tau
\big] \Big)^2 \Big] \;=\; 0\;.
\end{equation*}
This proves \eqref{l02}. Since $\Delta M_n$ is absolutely bounded by
one, in view of \eqref{l01} and \eqref{l02}, by \cite[Theorem
VIII.3.11]{JS}, the martingale $M_n(t)$ converges in distribution to
$B_{v(t)}$, where $B$ is a Brownian motion and $v(t) = \exp (- (\ln 2)
\, e^{-\beta \tau}) - (1/2)$. \qed

\end{document}